\documentclass[journal]{IEEEtran}
\usepackage{amsmath,amssymb,amsfonts}
\usepackage{algorithm}
\usepackage{array}
\usepackage{amsthm}
\usepackage[caption=false,font=normalsize,labelfont=sf,textfont=sf]{subfig}
\usepackage{textcomp}
\usepackage{stfloats}
\usepackage{url}
\usepackage{verbatim}
\usepackage{graphicx}
\usepackage{tabularx}
\usepackage{cite}
\usepackage{bm}

\usepackage{xcolor}
\usepackage{soul}
\usepackage{hyperref}
\usepackage{multirow}
\usepackage{multicol}
\usepackage{booktabs}
\usepackage[utf8x]{inputenc}
\usepackage{pgf}

\DeclareMathOperator*{\argmin}{argmin}

\newcommand{\mathbi}[1]{\bm{\mathit{#1}}}

\newcommand{\until}[1]{\mathbi{U}_{#1}}
\newcommand{\eventually}[1]{\Diamond_{#1}}
\newcommand{\always}[1]{\square_{#1}}

\newtheorem{theorem}{Theorem}
\newtheorem{corollary}{Corollary}

\newtheorem{proposition}{Proposition}

\newcommand{\norm}[1]{\left\lVert#1\right\rVert}

\newcommand{\real}{\mathbb{R}}
\newcommand{\integer}{\mathbb{Z}}
\newcommand{\z}{\mathbf{z}}
\newcommand{\x}{\mathbf{x}}
\renewcommand{\u}{\mathbf{u}}

\newcommand{\w}{\mathbf{w}}
\newcommand{\windvec}{\mathbf{v}}
\renewcommand{\v}{\mathbf{v}}

\newcommand{\rvec}{\mathbf{r}}
\newcommand{\tlleft}{[\![}
\newcommand{\tlright}{]\!]}
\newcommand{\tlleftopen}{(\!|}
\newcommand{\tlrightopen}{|\!)}

\newcommand{\windowlength}{\tau}

\begin{document}

\title{Motion Planning with Metric Temporal Logic Using Reachability Analysis and Hybrid Zonotopes}

\author{Andrew F. Thompson, Joshua A. Robbins, Jonah J. Glunt, Sean B. Brennan, and Herschel C. Pangborn
\thanks{Andrew F. Thompson, Joshua A. Robbins, Jonah J. Glunt, Sean B. Brennan, and Herschel C. Pangborn are with the Department of Mechanical Engineering, The Pennsylvania State University, University Park, PA 16802 USA (e-mail: {\tt\small thompson@psu.edu, jrobbins@psu.edu, jglunt@psu.edu, sbrennan@psu.edu, hcpangborn@psu.edu}).}
}

\maketitle

\begin{abstract}
Metric temporal logic (MTL) provides a formal framework for defining time-dependent mission requirements on autonomous vehicles. However, optimizing control decisions subject to these constraints is often computationally expensive. This article presents a method that uses reachability analysis to implicitly express the set of states satisfying an MTL specification and then optimizes to find a motion plan. The hybrid zonotope set representation is used to efficiently and conveniently encode MTL specifications into reachable sets. A numerical benchmark highlights the proposed method's computational advantages as compared to existing methods in the literature. Further numerical examples and an experimental application demonstrate the ability to address time-varying environments, region-dependent disturbances, and multi-agent coordination.
\end{abstract}

\section{Introduction}
Emerging applications of autonomous vehicles often require satisfaction of tasks or specifications beyond point-to-point navigation. Examples include self-driving cars~\cite{edviken2021safety}, package delivery drones~\cite{baran2021ros}, robots for inspection in dangerous or difficult to reach environments~\cite{silano2021power}, and urban air mobility vehicles~\cite{rodionova2020learning}, which all come with complex safety constraints and performance objectives that must be considered in path planning. Temporal logic (TL) provides a formal representation by which to express these requirements as time-dependent constraints. Metric temporal logic (MTL), signal temporal logic (STL), and other extensions of linear temporal logic (LTL) are particularly well-suited for autonomous systems, as they incorporate an explicit notion of time into their libraries of operators.
A key challenge in the use of TL for motion planning is to encode these specifications in a computational framework that facilitates efficient and scalable solution of compliant trajectories. 

\subsection{Gaps in the Literature}
Automata-based methods are commonly employed for control synthesis under TL specifications~\cite{Belta2019,Kurtz2023,Kalluraya2023}. Typically, this involves solving a graph search in the product of a finite-state abstraction of the system's state space with a finite-state automaton representing the TL specification~\cite{Belta2019, Kurtz2023}. While these methods can be efficient to implement for certain specifications~\cite{Belta2017}, the abstraction of system dynamics and TL specifications might yield a trajectory that is infeasible or suboptimal. Other automaton-based approaches have overcome these limitations by using graphs of convex sets to solve for optimal trajectories under LTL specifications~\cite{Kurtz2023}.

Another popular approach is to use optimization-based methods which incorporate a mixed-integer program (MIP). These methods do not require a state-space abstraction, so the resulting trajectory is dynamically feasible. MIPs have been used for control under LTL~\cite{Wolff2014, Luo2022}, MTL~\cite{Kurtz2022a, zhou2015optimal}, STL~\cite{Kurtz2022, Yu2023, Buyukkocak2021, Belta2019}, and other variations of TL~\cite{Cardona2023}. In these MIPs, binary variables and additional constraints track and enforce the desired TL specification. Optimization approaches are guaranteed to converge to a global optimum at a potentially great computational cost~\cite{Belta2019}. The worst-case performance of MIPs grows exponentially with respect to the number of binary variables and is also greatly affected by the tightness of the formulation's convex relaxation (i.e., how much the feasible space of the optimization program grows when binary variables are relaxed to take continuous values). The standard approaches for encoding TL specifications in MIPs introduce many binary variables and use the so-called ``\mbox{big M}'' method that produces notoriously loose convex relaxations and can hinder branch and bound solvers~\cite{Bonami2015}.

Finally, there are set-based methods that use reachability analysis of dynamic systems to verify or falsify temporal logic specifications, as applied to autonomous driving in~\cite{Arfvidsson2024, Arfvidsson2024towards, Jiang2024, Hadjiloizou2024}. Some of these approaches include a combination of a temporal logic tree with backwards reachability of the system, beginning at the atomic propositions of the specification and terminating when the set reaches invariance~\cite{Jiang2024,Hadjiloizou2024}. These methods either use Hamilton-Jacobi reachability~\cite{Jiang2024}, where exact solutions can scale exponentially with the state dimension~\cite{bansal2017hamilton} and scalable algorithms can produce approximate solutions~\cite{sharpless2023koopman}, or require complements of hybrid zonotope set representations~\cite{Hadjiloizou2024}, which introduce significant computational complexity~\cite{Bird2022}. Hybrid zonotopes have also been used to represent a subset of STL, leveraging graphs of convex sets specifically for door and key maze-like problems~\cite{you2025}.

\subsection{Contributions}
This article presents an approach for autonomous vehicle motion planning under MTL specifications. The hybrid zonotope set representation is leveraged to compute forward reachable sets that satisfy the system dynamics, state and input constraints, and MTL specifications. These sets describe the feasible space of a MIP, which is solved to generate a motion plan. 
The key contribution is a method for encoding MTL specifications within hybrid zonotopes using straightforward set identities that do not rely on the big M method. This often leads to an MIP of lower complexity (e.g., an order of magnitude fewer binary variables) than state-of-the-art approaches for expressing MTL specifications as MIPs. 
Additionally, we introduce a method to construct constraints that represent the temporal ``until" operator with fewer disjunctions than in the literature. 
The proposed approach is demonstrated in a variety of numerical and experimental examples. Benchmarking with a door-key problem from the literature highlights its improved representational and computational performance. Further examples highlight the ability to address time-varying environments, region-dependent disturbances, and multi-agent coordination.

\subsection{Outline}
The remainder of this paper is organized as follows. Sec.~\ref{sec:prelims} provides preliminary material, including notation, set representations, and the syntax of MTL. Sec.~\ref{sec:methods} describes how we represent reachable sets and obstacle/goal maps using hybrid zonotopes, and presents the main results for encoding MTL specifications in hybrid zonotope reachable sets. Sec.~\ref{sec:examples} provides numerical examples to demonstrate and benchmark the proposed method. An experimental implementation is given in Sec.~\ref{sec:experiment} and concluding remarks are provided in Sec.~\ref{sec:conclusion}.

\section{Preliminaries} \label{sec:prelims}
\subsection{Notation}
Vectors are denoted with boldface letters. Matrices are denoted with capital letters. The identity matrix of size $n$ is denoted by $I_n$. Sets are denoted with calligraphic letters or with curly braces $\{\cdot\}$. The interior of a set $\mathcal{X}$ 
is denoted as $\mathcal{X}^\circ$. The power set, or the set of all possible subsets, of $\mathcal{X}$ is denoted as $2^{\mathcal{X}}$. The empty set is denoted as $\emptyset$. Brackets $(a,b)$ and $[a,b]$ denote the open and closed intervals from $a$ to $b$, respectively. Double brackets $\tlleftopen a, b \tlrightopen$ and $\tlleft a, b \tlright$ denote the open and closed intervals of integers from $a$ to $b$, respectively, i.e., $\tlleft a, b \tlrightopen \triangleq [a, b) \cap \integer$. The squared norm of a vector with respect to a matrix is $\norm{\x}^2_A \triangleq \x^T A \x$. 

\subsection{Set Representations and Operations} \label{sec:hybzono-definition}
A set $\mathcal{Z} \subset \real^n$ is a zonotope if $\exists \; G^c \in \real^{n \times n_g}$ and $\bm{c} \in \real^{n}$ such that
\begin{equation} \label{eq:zonotope}
\mathcal{Z} = \left\{ G^c \bm{\xi}^c + \bm{c} \mid \bm{\xi}^c \in [0,1]^{n_g} \right\} \;.
\end{equation}
Zonotopes are convex, centrally symmetric sets \cite{ziegler2012lectures}.

A set $\mathcal{Z}_C \subset \real^n$ is a constrained zonotope if $\exists \; G^c \in \real^{n \times n_g}$, $\bm{c} \in \real^{n}$, $A^c \in \real^{n_c \times n_g}$, and $\bm{b} \in \real^{n_c}$ such that
\begin{equation} \label{eq:cons_zonotope}
\mathcal{Z}_C = \left\{ G^c \bm{\xi}^c + \bm{c} \mid \bm{\xi}^c \in [0,1]^{n_g},\; A^c \bm{\xi}^c = \bm{b} \right\} \;.
\end{equation} 
Constrained zonotopes can represent any convex polytope \cite{scott2016constrained}.

Hybrid zonotopes extend \eqref{eq:cons_zonotope} by including binary factors $\bm{\xi}^b$. A set $\mathcal{Z}_H \subset \real^n$ is a hybrid zonotope if in addition to $G^c$, $\bm{c}$, $A^c$, and $\bm{b}$, $\exists\; G^b \in \real^{n \times n_b}$ and $A^b \in \real^{n_c \times n_b}$ such that
\begin{equation} \label{eq:hyb_zonotope}
\mathcal{Z}_H = \left\{ \begin{bmatrix} G^c & G^b \end{bmatrix} 
\begin{bmatrix} \bm{\xi}^c \\ \bm{\xi}^b \end{bmatrix} + \bm{c} \;\middle|\; 
\begin{matrix}
\begin{bmatrix} \bm{\xi}^c \\ \bm{\xi}^b \end{bmatrix} \in \mathcal[0,1]^{n_g} \times \{0,1\}^{n_b} \\
\begin{bmatrix} A^c & A^b \end{bmatrix} \begin{bmatrix} \bm{\xi}^c \\ \bm{\xi}^b \end{bmatrix} = \bm{b}
\end{matrix}
\right\} \;.
\end{equation}
Hybrid zonotopes can represent any union of polytopes \cite{bird2023hybrid}. 
While hybrid zonotopes were initially defined with generators ranging from $-1$ to $1$, in this paper we consider an equivalent representation that has generators ranging from $0$ to $1$, as in~\cite{robbins2024efficientjournal, Thompson2025}. Hybrid zonotopes are denoted using the shorthand notation $\mathcal{Z}_H = \left\langle G^c, G^b, \bm{c}, A^c, A^b, \bm{b} \right\rangle$.

A set $\mathcal{H} \subset \real^n_x$ is a halfspace representation (H-rep) polytope if it is bounded and $\exists \; L \in \real^{n_i \times n_x}$, $\bm{r} \in \real^{n_i}$, $A \in \real^{n_e \times n_x}$, and $\bm{b} \in \real^{n_e}$ such that 
\begin{equation} \label{eq:hrep}
    \mathcal{H} = \left\{ \x \in \real^{n_x} \mid L \x \leq \bm{r}\,, A \x = \bm{b}\right\}\,.
\end{equation}
This definition of a H-rep polytope includes an equality constraint for convenience in the main results. It is trivial to show that these equality constraints can be converted to a pair of inequality constraints, yielding a more conventional polytope definition. H-rep polytopes are simple to construct and have efficient identities for set intersection with other H-rep polytopes~\cite{Althoff2021} as well as with hybrid zonotopes~\cite{bird_dissertation}. H-rep polytopes are denoted using the shorthand notation $\mathcal{H} = \left\langle L, \bm{r}, A, \bm{b} \right\rangle_h$.

One major advantage of hybrid zonotopes and constrained zonotopes as compared to other polytopic set representations is that they have closed form and efficient expressions for key set operations. Given the sets $\mathcal{X}$, $\mathcal{Y} \subset \real^n$ and $\mathcal{Z} \subset \real^m$, the following set operations are defined as
\begin{subequations}
\begin{align}
    R \mathcal{X} + \mathbf{c} &= \left\{R \mathbf{x} + \mathbf{c} \;\middle|\; \mathbf{x} \in \mathcal{X}\right\} \label{eq:set-ops-affine}\,,\\
    \mathcal{X} \times \mathcal{Z} &= \left\{\begin{bmatrix}
        \mathbf{x} \\ \mathbf{z} 
    \end{bmatrix}\;\middle|\; \mathbf{x} \in \mathcal{X},\; \mathbf{z} \in \mathcal{Z}\right\} \label{eq:set-ops-cart-prod}\,,\\
    \mathcal{X} \oplus \mathcal{Y} &= \left\{\mathbf{x} + \mathbf{y} \;\middle|\; \mathbf{x} \in \mathcal{X},\; \mathbf{y} \in \mathcal{Y}\right\} \label{eq:set-ops-mink-sum}\,,\\
    \mathcal{X} \cap_R \mathcal{Z} &= \left\{\mathbf{x} \in \mathcal{X} \;\middle|\;  R \mathbf{x} \in \mathcal{Z}\right\} \label{eq:set-ops-intersection}\,,
\end{align}
\end{subequations}
where~\eqref{eq:set-ops-affine} is the affine map,~\eqref{eq:set-ops-cart-prod} is the Cartesian product,~\eqref{eq:set-ops-mink-sum} is the Minkowski sum, and~\eqref{eq:set-ops-intersection} is the generalized intersection. The representational complexities and time complexities of performing these operations in the case where $\mathcal{X}$, $\mathcal{Y}$, and $\mathcal{Z}$ are hybrid zonotopes is summarized in~\cite[Table 3.1]{bird_dissertation}. 
\begin{proposition}
The generalized intersection of a hybrid zonotope and an H-rep polytope as defined in~\eqref{eq:hrep} is given by
\begin{align}
\begin{split}
    &\mathcal{Z}_H \cap_R \mathcal{H} = \left<\begin{bmatrix}
        G_z^c & \mathbf{0}_{n \times n_i}
    \end{bmatrix}, G_z^b, \bm{c}, \right.\\
    &\left.\begin{bmatrix}
        A_z^c & \mathbf{0}_{n_c \times n_i}\\ L_h R G_z^c & \mathtt{diag}(\bm{s}) \\ A_h R G_z^c & \mathbf{0}_{n_e \times n_i}
    \end{bmatrix}, \begin{bmatrix}
        A_z^b \\ L_h R G_z^b \\ A_h R G_z^b
    \end{bmatrix}, \begin{bmatrix}
        \bm{b}_z \\ \bm{r}_h - L_h R \bm{c}_z\\ \bm{b}_h
    \end{bmatrix}\right> \,,
\end{split}
\end{align}
where 
\begin{equation}
    \bm{s} = \bm{r}_h - L_h R c_z + \sum_{i=1}^{n_g} \left|L_h R g_z^{c,i}\right| + \sum_{i=1}^{n_b} \left|L_h R g_z^{b,i}\right|\,.
\end{equation}    
and where $g_z^{c,i}$ and $g_z^{b,i}$ are the $i^{\text{th}}$ columns of $G^c$ and $G^b$, respectively.
\end{proposition}
\begin{proof}
This is given from the intersection identity in~\cite[Prop. 3.2.4]{bird_dissertation} and the conversion of hybrid zonotopes from $\{-1,1\}$ to $\{0,1\}$ basis in~\cite[Prop. 1]{robbins2024efficientjournal}. This is also an extension of~\cite[Prop. 1]{Rego2024} from constrained zonotopes to hybrid zonotopes.
\end{proof}

\subsection{Metric Temporal Logic} \label{sec:prelim-mtl}
MTL extends LTL by incorporating an explicit notion of time into its operators~\cite{Koymans1990,Belta2019, bellini2000temporal, plaku2015motion}. In the case of discrete-time systems, ``time" refers to a time index (easily related to real time through discretization). The satisfaction or dissatisfaction of MTL specifications take binary values. 

The syntax of MTL is defined as
\begin{equation}
    \varphi \triangleq \begin{array}{c|c|c|c|c}
        \top & \pi & \neg\varphi & \varphi_1 \wedge \varphi_2 & \varphi_1 \until{\tlleft t_1, t_2 \tlright} \varphi_2\;,
    \end{array}
\end{equation}
where $\top$ denotes ``true", $\pi \in \mathcal{P}$ is an atomic proposition from a set of predetermined propositions, $\neg$ is the ``negation" operator, $\wedge$ is the boolean ``and" operator, and $\until{\tlleft t_1, t_2 \tlright}$ is the ``until" operator. The until operator requires that $\varphi_1$ remain true until $\varphi_2$ becomes true at some point in the interval $\tlleft t_1, t_2 \tlright$. Additional useful operators can be derived from these base operators~\cite{Ouaknine2005}. For instance, the boolean ``or" ($\varphi_1 \vee \varphi_2$), temporal ``eventually" ($\eventually{\tlleft t_1, t_2 \tlright} \varphi$), and temporal ``always" ($\always{\tlleft t_1, t_2 \tlright} \varphi$) operators can be defined as 
\begin{subequations}
\begin{align}
    \varphi_1 \vee \varphi_2 &\triangleq \neg\left(\neg \varphi_1 \wedge \neg \varphi_2\right)\;,\\
    \eventually{\tlleft t_1, t_2 \tlright} \varphi &\triangleq \top \until{\tlleft t_1, t_2 \tlright} \varphi\;, \label{eq:eventually-def}\\
    \always{\tlleft t_1, t_2 \tlright} \varphi &\triangleq \neg \eventually{\tlleft t_1, t_2 \tlright} \neg \varphi\;. \label{eq:always-def}
\end{align}
\end{subequations}
The satisfaction of these operators is evaluated on a finite trajectory of states $\bm{\sigma} = \left[\mathcal{T}_0, \mathcal{T}_1, \dots, \mathcal{T}_N\right]$ where $\mathcal{T}_k \in 2^\mathcal{P}$ are the set of propositions that are true at time step $k$. These operators are mathematically defined as
\begin{subequations}\label{eq:mtl-def}
\begin{align}
    \bm{\sigma} \models_k \pi & \Leftrightarrow \pi \in \mathcal{T}_k\,,\label{eq:mtl-def-prop}\\
    \bm{\sigma} \models_k \neg \varphi & \Leftrightarrow\bm{\sigma} \not\models_k \varphi\,,\label{eq:mtl-def-neg}\\
    \bm{\sigma} \models_k \varphi_1 \wedge \varphi_2 & \Leftrightarrow\bm{\sigma} \models_k \varphi_1 \text{ and }\bm{\sigma} \models_k \varphi_2\,,\label{eq:mtl-def-and}\\
    \bm{\sigma} \models_k \varphi_1 \vee \varphi_2 & \Leftrightarrow \bm{\sigma} \models_k \varphi_1 \text{ or }\bm{\sigma} \models_k \varphi_2\,,\label{eq:mtl-def-or}\\
    \bm{\sigma} \models_k \varphi_1 \until{\tlleft t_1, t_2 \tlright} \varphi_2 & \Leftrightarrow \exists t \in \tlleft k + t_1, k + t_2 \tlright \text{ s.t. } \bm{\sigma} \models_t \varphi_2\nonumber\\
    & \quad \text{and } \forall t' \in \tlleft k + t_1,t\tlrightopen, \bm{\sigma} \models_{t'} \varphi_1 \,,\label{eq:mtl-def-until}\\
    \bm{\sigma} \models_k \eventually{\tlleft t_1, t_2 \tlright}\varphi & \Leftrightarrow\exists t \in \tlleft k +t_1, k+t_2 \tlright \text{ s.t. } \bm{\sigma} \models_t \varphi\,,\label{eq:mtl-def-eventually}\\
    \bm{\sigma} \models_k \always{\tlleft t_1, t_2 \tlright}\varphi & \Leftrightarrow  \forall t \in \tlleft k +t_1, k+t_2 \tlright, \bm{\sigma} \models_t \varphi\,,\label{eq:mtl-def-always}
\end{align}
\end{subequations}
where $\bm{\sigma} \models_k \varphi$ denotes $\mathcal{T}_k$ satisfying specification $\varphi$.
This definition of the until operator is sightly different than convention, as $\varphi_1$ only needs to hold true beginning at $t=k+t_1$, as opposed to $t = k$. This is not restrictive as the more conventional until operator, $\varphi_1 \until{\tlleft t_1, t_2 \tlright}' \varphi_2$, can be reconstructed as 
\begin{equation}
    \varphi_1 \until{\tlleft t_1, t_2 \tlright}' \varphi_2 = \left(\always{\tlleft 0, t_1-1 \tlright} \varphi_1\right) \wedge \left(\varphi_1 \until{\tlleft t_1, t_2 \tlright} \varphi_2\right) \, .
\end{equation}
\section{Encoding MTL Specifications in Reachable Sets}\label{sec:methods}
This section presents the proposed method for path planning of dynamic systems to satisfy a desired MTL specification.
First, we mathematically state the problem to be solved. Then we construct a set that relates the system's state to the satisfaction of various propositions, which we refer to as a ``map." We then iteratively construct the set of all dynamically feasible trajectories using forward reachability analysis and represent these reachable sets across all time steps within a single, high-dimensional hybrid zonotope. We then construct sets that represent the MTL operators described in~\eqref{eq:mtl-def} and enforce them through set intersections. Finally, we optimize over this high-dimensional set to produce an optimal trajectory.

\subsection{Problem Statement}
This paper considers the problem of motion planning for autonomous systems subject to MTL specifications. We consider a discrete-time linear time-invariant (LTI) dynamic model given by
\begin{equation} \label{eq:lti}
    \x_{k+1} = A \x_k + B \u_k + \w_k\,,
\end{equation}
where $\x_k \in \mathcal{X}_k \subset \real^{n_x}$ is the state vector at time step $k$, $\u_k \in \mathcal{U}_k \subset \real^{n_u}$ is the control input, and $\w_k \in \mathcal{W}_k \subset \real^{n_w}$ is a state-dependent disturbance. The latter is primarily used to capture spatially-varying disturbances such as wind for an aerial vehicle or surface grade for a ground vehicle, which we assume to be known \textit{a priori}. 

In this work, we consider the proposition of our MTL formulae to indicate whether or not the system's state lies within a polytopic region of the state space, $\mathcal{X}_{\pi}$. It becomes useful to create a ``map," $\mathcal{M} \subset \real^{n_x}$, that translates between the system's state and the satisfaction of propositions. Using a hybrid zonotope to define this map, we can assign one binary variable for each region, i.e., 
\begin{equation} \label{eq:state_in_roi}
    \bm{\sigma} \models_k \pi \Leftrightarrow \x_k \in \mathcal{X}_\pi \Leftrightarrow \xi^b_{\pi,k} = 1\,.
\end{equation}
We then encode the desired MTL specification as a collection of linear inequalities on the system's states, from which we can construct the feasible space of each MTL formula as H-rep polytopes. Taking the generalized intersection of these sets with the reachable set yields a set that contains all dynamically feasible trajectories which also satisfy the desired MTL specification. Lastly, we construct and solve a mixed-integer quadratic program over an $N$-step horizon the form
\begin{subequations} \label{eq:generic-lqr}
\begin{align}
    \u^* = \argmin_{\u}& \norm{\x_N}^2_{Q_N} + \sum_{k=0}^{N-1} \norm{\x_k}^2_Q + \norm{\u_k}^2_R\,, \\
    \text{s.t. }& \forall k \in \{0,\dots,N-1\}\,,\\
    &\x_k \in \mathcal{X}_k\,,\; \u_k \in \mathcal{U}_k\,,\x_N \in \mathcal{X}_N\,,\\
    &\x_{k+1} = A \x_k + B \u_k + \w_k\,,\\
    &\x \models_0 \varphi \,,
\end{align}
\end{subequations}
where $Q_N$, $Q$, and $R$ are the cost matrices associated with the terminal state, stage states, and control inputs, respectively, and $\mathcal{U}_k$, $\mathcal{X}_k$, and $\mathcal{X}_N$ are represented as hybrid zonotopes.

\subsection{Map Representation} \label{sec:map_construction}

For a given MTL specification and set of propositions, we create a map that relates regions of the state space to each proposition. The map is a union of convex polytopes and represented using a hybrid zonotope of the form 
\begin{equation} \label{eq:map-set}
    \mathcal{M} = \left<G^c, G^b, c, A^c, A^b, b\right> \subset \real^{n_x}\,,
\end{equation}
via the union identity in~\cite[Thm. 5]{Siefert2023tac}, which introduces one binary factor per polytope. Each polytope could represent an obstacle, keep-out zone, region of interest, or an otherwise unassigned region of the state space. 
We then augment this hybrid zonotope to include its binary factors into the state vector. The augmented map, $\mathcal{M}'$, can be constructed by the generator matrices $G^{c'}$, $G^{b'}$, and center $c'$ defined as
\begin{equation}\label{eq:hz-augment}
    G^{c'} = \begin{bmatrix}
        G^c \\ \mathbf{0}
    \end{bmatrix}\, ,\; 
    G^{b'} = \begin{bmatrix}
        G^b \\ I_{n_b}
    \end{bmatrix}\, ,\;
    c' = \begin{bmatrix}
        c \\ \mathbf{0}
    \end{bmatrix}\, .    
\end{equation}
With $\mathcal{M}' = \left<G^{c'}, G^{b'}, c', A^c, A^b, b\right> \subset \real^{n_x} \times \{0,1\}^{n_b}$, the binary variables associated with the state vector at a given time step can be accessed, i.e., $[\x_k^T, \bm{\xi}^{bT}_k]^T \in \mathcal{M}'$. For cases where the map is time varying, we denote the augmented map at time step $k$ as $\mathcal{M}'_k$. This exposure of the binary variables allows the construction of temporal constraints along these dimensions and the incorporation of region-dependent disturbances. To enforce these disturbances, we assign a disturbance vector $\v_i$ associated with each polytopic region of the map, collect them into a disturbance matrix $W = [\v_1,\dots,\v_{n_b}]$, and use a linear transformation $\w_k = W \bm{\xi}^b_k$ to define the disturbance at each time step. Since $\bm{\xi}^b_k \in \{0,1\}^{n_b}$ and only one element is non-zero (say the $i^{\text{th}}$ element, corresponding to $\x_k$ being in the $i^{\text{th}}$ region of the map), the linear transformation produces $W \bm{\xi}^b_k = \v_i$.

We present two ways to partition the vehicle's environment into polytopic regions:

\subsubsection{Disjoint Convex Partition}
For a state set $\mathcal{X}$, a disjoint convex partition (DCP) is a collection of $n$ convex polytopes $\{\mathcal{X}_{\pi,1}, \dots, \mathcal{X}_{\pi,n}\}$ such that
\begin{subequations}
\begin{align}
    \bigcup_{j=1}^n \mathcal{X}_{\pi,j} &= \mathcal{X}\;, \label{eq:union_eqs_globl}\\
    \mathcal{X}_{\pi,j} \cap \mathcal{X}^\circ_{\pi,k} &= \emptyset,\; \forall j,k \in \{1,\dots,n\},\; j\neq k,
\end{align}
\end{subequations}
i.e., each subregion at most shares a border with any other subregion. This definition of a DCP enforces that the system's state is associated with exactly one proposition, with the exception of the boundaries between propositions. At these boundaries, the state can satisfy any one of the bordering propositions, therefore a DCP is only able to prohibit strict obstacle entry. Using a DCP is necessary when optimizing trajectories through environments with obstacles or keep-out zones. If we permitted a target region $\mathcal{X}_{\pi,1}$ to overlap an obstacle $\mathcal{X}_{\pi,2}$, then it would be feasible for the state to occupy the portion of $\mathcal{X}_{\pi,1}$ that intersects $\mathcal{X}_{\pi,2}$. The Hertel and Mehlhorn algorithm~\cite{hertel_mehlhorn} can be used to create DCPs for two-dimensional environments. Fig.~\ref{fig:example-map}(a) shows an example of a DCP around a polytopic region $\mathcal{X}_{\pi,1}$ using six additional polytopes.

\subsubsection{Non-disjoint Convex Partition}
For a state set $\mathcal{X}$, a non-disjoint convex partition (NCP) is a collection of $n$ convex polytopes $\{\mathcal{X}_{\pi,1}, \dots, \mathcal{X}_{\pi,n}\}$ satisfying only~\eqref{eq:union_eqs_globl}. Therefore, one convex polytope of the partition might intersect or be completely contained by another. The state can then satisfy a proposition associated with any of the overlapping polytopes. NCPs are useful when there are no obstacles or keep-out zones to be considered, or when combining target regions with a previously constructed DCP (as in Sec.~\ref{sec:example-door-key}), as we can represent the map with fewer convex regions than a DCP. Fig.~\ref{fig:example-map}(b) shows an example of an NCP, with a polytopic region inside a rectangular environment where $\mathcal{X}_{\pi,2}$ can exist in the ``background" of $\mathcal{X}_{\pi,1}$.

\begin{figure}[tb]
    \centering
    \input{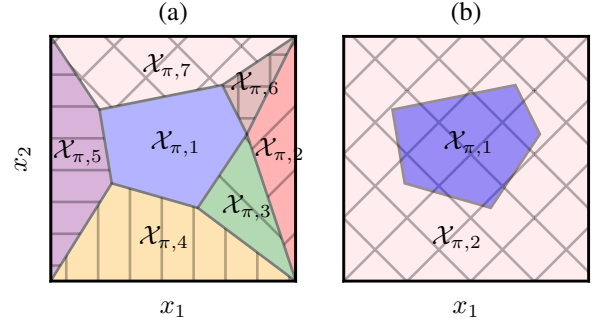}
    \caption{Example of a rectangular environment with a single region of interest, $\mathcal{X}_{\pi,1}$, created using (a) a disjoint convex partition and (b) a non-disjoint convex partition.}
    \label{fig:example-map}
\end{figure}

\subsection{Reachability Analysis} \label{sec:reach}

We now create a high-dimensional hybrid zonotope that represents all dynamically feasible trajectories of the system described in~\eqref{eq:lti}. Given a set of initial lifted states $\mathcal{X}'_0 = \left\{[\x^T, \bm{\xi}^{bT}]^T \middle | \x \in \mathcal{X}_0, \bm{\xi}^b \in\{0,1\}^{n_b}\right\}\subseteq \mathcal{M}'_0$, control inputs $\mathcal{U}_k$, disturbance matrix $W\in\real^{n \times n_b}$, and maps $\mathcal{M}_k$, forward reachability analysis can be performed via the recursive set operation
\begin{equation}\label{eq:reach-naive}
    \mathcal{X}'_{k+1} = \left(\begin{bmatrix}
        A & W \\ \mathbf{0} & \mathbf{0}
    \end{bmatrix} \mathcal{X}'_k \oplus \begin{bmatrix}
        B  \\  \mathbf{0}
    \end{bmatrix} 
    \mathcal{U}_k \oplus \begin{bmatrix}
        \mathbf{0} \\ I_{n_b}
    \end{bmatrix}\{0,1\}^{n_b} \right) \cap \mathcal{M}'_{k+1}\,.
\end{equation}

To produce reachable sets that are equivalent to~\eqref{eq:reach-naive} but have a structure that is better suited to sparsity-exploiting optimization, we instead employ the identity for forward reachability described in~\cite{robbins2025sparsitypromotingreachabilityanalysisoptimization}, given as
\begin{equation}
    \mathcal{X}'_{k+1} = \begin{bmatrix}
        \mathbf{0} & \mathbf{0} & I_{n_x+n_b}
    \end{bmatrix} \left( \left(\mathcal{X}'_k \times \mathcal{U}_k \times \mathcal{M}'_{k+1}\right) \cap_R\left\{\mathbf{0}\right\}\right)\,,
\end{equation}
where
\begin{equation}
    R = \begin{bmatrix}
        A & W & B & -I_{n_x} & \mathbf{0}_{n_x \times n_b}
    \end{bmatrix} \;.
\end{equation}
We can collect these reachable sets, along with the allowable input sets, to create one lifted set that contains all dynamically feasible trajectories over the horizon as
\begin{equation}
    \mathcal{Z}_{\text{reach}} = \mathcal{X}'_0 \times \mathcal{U}_0 \times \dots \times \mathcal{U}_{N-1} \times \mathcal{X}'_N\,.
\end{equation}
Therefore, any specific feasible trajectory of states, binary generators associated with those states, and control inputs are contained by this lifted reachable set, i.e.,
\begin{equation}
    \begin{bmatrix}
        \x_0^T & \bm{\xi}^{bT}_0 & \u_0^T & \dots & \u_{N-1}^T & \x_N^T & \bm{\xi}^{bT}_N
    \end{bmatrix}^T \in \mathcal{Z}_{\text{reach}}\,.
\end{equation}
The memory complexity of representing $\mathcal{Z}_{\text{reach}}$ is
\begin{subequations}
    \begin{align}
        n^{\mathcal{Z}_{\text{reach}}} &= \sum_{k=0}^{N-1} \left(n^{\mathcal{X}'_k} + n^{\mathcal{U}_k}\right) + n^{\mathcal{X}'_N}\,,\\
        n_g^{\mathcal{Z}_{\text{reach}}} &= \sum_{k=0}^{N-1} \left(n_g^{\mathcal{X}'_k} + n_g^{\mathcal{U}_k}\right) + n_g^{\mathcal{X}'_N}\,,\\
        n_b^{\mathcal{Z}_{\text{reach}}} &= \sum_{k=0}^{N-1} \left(n_b^{\mathcal{X}'_k} + n_b^{\mathcal{U}_k}\right) + n_g^{\mathcal{X}'_N}\,,\\
        n_c^{\mathcal{Z}_{\text{reach}}} &= \sum_{k=0}^{N-1} \left(n_c^{\mathcal{X}'_k} + n_c^{\mathcal{U}_k} + n^{\mathcal{X}_{k}}\right) + n_c^{\mathcal{X}'_N}\,.
    \end{align}
\end{subequations}

\subsection{Representation of Key MTL Operators as H-rep Polytopes}\label{sec:mtl-to-hrep}
This section presents methods for constructing the feasible space of all key MTL operators as H-rep polytopes.
\subsubsection{Boolean Operators}
The boolean operators presented in~\eqref{eq:mtl-def-prop}-\eqref{eq:mtl-def-or} enforce validity of certain propositions, or pairs of propositions, at a specific time index $k$. These operators can be constructed using linear equalities and inequalities.
\begin{proposition}
Given a map as described in Sec.~\ref{sec:map_construction} with propositions $\pi$, $\varphi_1$, and $\varphi_2$ where $\bm{\sigma} \models_k \varphi \Leftrightarrow \xi^b_{\varphi,k} = 1$, the MTL specifications for proposition satisfaction, boolean negation, ``and", and ``or" can be represented with through a single linear inequality or equality constraint of the form
\begin{subequations}\label{eq:boolean-ineqs}
\begin{align} 
    \bm{\sigma} \models_k \pi & \Leftrightarrow \xi^b_{\pi, k} = 1\,, \label{eq:prop-eq}\\
    \bm{\sigma} \models_k \neg \varphi_1 & \Leftrightarrow \xi^b_{\varphi_1, k} = 0\,, \label{eq:neg-eq}\\
    \bm{\sigma} \models_k \varphi_1 \wedge \varphi_2 & \Leftrightarrow \xi^b_{\varphi_1, k} + \xi^b_{\varphi_2, k} = 2\,, \label{eq:and-eq}\\
    \bm{\sigma} \models_k \varphi_1 \vee \varphi_2 & \Leftrightarrow \xi^b_{\varphi_1, k} + \xi^b_{\varphi_2, k} \geq 1\,. \label{eq:or-ineq}
\end{align}
\end{subequations}
\end{proposition}
\begin{proof}
    For proposition satisfaction,~\eqref{eq:prop-eq} is a stated assumption and is included for completeness. For negation, it follows from~\eqref{eq:mtl-def-neg},~\eqref{eq:prop-eq}, and $\xi^b_k \in \{0,1\}$ that $\xi^b_{\varphi_k} = 0 \rightarrow \xi^b_{\varphi_k} \neq 1 \rightarrow \bm{\sigma} \not\models_k \varphi \rightarrow \bm{\sigma} \models_k \neg \varphi$ and that $\bm{\sigma} \models_k \neg \varphi \rightarrow \bm{\sigma} \not\models_k \varphi \rightarrow \xi^b_{\varphi_k} \neq 1 \rightarrow \xi^b_{\varphi_k} = 0$. For ``and", it follows from~\eqref{eq:mtl-def-and} and $\xi^b_{\varphi_1, k}, \xi^b_{\varphi_2, k} \in \{0,1\}$ that $\xi^b_{\varphi_1, k} + \xi^b_{\varphi_2, k} = 2 \rightarrow \xi^b_{\varphi_1, k} = 1$ and $\xi^b_{\varphi_2, k} = 1 \rightarrow \bm{\sigma} \models_k \varphi_1$ and $\bm{\sigma} \models_k \varphi_1 \rightarrow \bm{\sigma} \models_k \varphi_1 \wedge \varphi_2$. Conversely, $\bm{\sigma} \models_k \varphi_1 \wedge \varphi_2 \rightarrow \bm{\sigma} \models_k \varphi_1$ and $\bm{\sigma} \models_k \varphi_1 \rightarrow \xi^b_{\varphi_1, k} = 1$ and $\xi^b_{\varphi_2, k} = 1 \rightarrow \xi^b_{\varphi_1, k} + \xi^b_{\varphi_2, k} = 2$. For ``or", it follows from~\eqref{eq:mtl-def-or} and $\xi^b_{\varphi_1, k}, \xi^b_{\varphi_2, k} \in \{0,1\}$ that $\xi^b_{\varphi_1, k} + \xi^b_{\varphi_2, k} \geq 1 \rightarrow \xi^b_{\varphi_1, k} = 1$ or $\xi^b_{\varphi_2, k} = 1 \rightarrow \bm{\sigma} \models_k \varphi_1$ or $\bm{\sigma} \models_k \varphi_1 \rightarrow \bm{\sigma} \models_k \varphi_1 \vee \varphi_2$. Conversely, $\bm{\sigma} \models_k \varphi_1 \vee \varphi_2 \rightarrow \bm{\sigma} \models_k \varphi_1$ or $\bm{\sigma} \models_k \varphi_1 \rightarrow \xi^b_{\varphi_1, k} = 1$ or $\xi^b_{\varphi_2, k} = 1 \rightarrow \xi^b_{\varphi_1, k} + \xi^b_{\varphi_2, k} \geq 1$.
\end{proof}

To enforce these constraints on reachable sets, we construct H-rep polytopes that contain the feasible space of the boolean operators.

\begin{proposition} \label{prop:boolean-hrep}
    The feasible space of the constraints in~\eqref{eq:boolean-ineqs} can be represented as H-rep polytopes of the form
    \begin{subequations} \label{eq:boolean-cz}
    \begin{align}
        \mathcal{H}_{\pi} &= \left<[\,], [\,], 1, 1\right>_h\subset \real\,,\\
        \mathcal{H}_{\neg} &= \left<[\,], [\,], 1, 0\right>_h\subset \real\,,\\
        \mathcal{H}_{\wedge} &= \left<[\,], [\,], \begin{bmatrix}
            1 & 1
        \end{bmatrix}, 2\right>_h\subset \real^{2}\,,\\
        \mathcal{H}_{\vee} &= \left<\begin{bmatrix}
            -1 & -1
        \end{bmatrix}, -1, [\,], [\,]\right>_h\subset \real^{2} \label{eq:or-cz}\,,
    \end{align}
    \end{subequations}
    where $\xi^b_{\pi,k} \in \mathcal{H}_{\pi}$, $\xi^b_{\varphi,k} \in \mathcal{H}_{\neg}$, $[\xi^b_{\varphi_1,k}, \xi^b_{\varphi_2,k}]^T \in \mathcal{H}_{\wedge}$, and $[\xi^b_{\varphi_1,k}, \xi^b_{\varphi_2,k}]^T \in \mathcal{H}_{\vee}$.
\end{proposition}

\begin{proof}
    By definition of an H-rep polytope in~\eqref{eq:hrep}, $\mathcal{H}_{\pi} = \left\{x\in \real \;\middle|\; x = 1\right\}$, $\mathcal{H}_{\neg} = \left\{x\in \real \;\middle|\; x = 0\right\}$, $\mathcal{H}_{\wedge} = \left\{\x\in \real^2 \;\middle|\; \x_1 + \x_2 = 2\right\}$, $\mathcal{H}_{\vee} = \left\{\x\in \real^2 \;\middle|\; -\x_1 - \x_2 \leq -1\right\}$. Thus, any variables that satisfy the constraints in~\eqref{eq:boolean-ineqs} belong to the feasible space of~\eqref{eq:boolean-cz}.
\end{proof}

\subsubsection{Until}
As explained in Sec.~\ref{sec:prelim-mtl}, the until operator, $\varphi_1 \until{\tlleft t_1, t_2 \tlright} \varphi_2$,  requires that the first operand, $\varphi_1$, remains true beginning at $t_1$ and until the second operand, $\varphi_2$, becomes true at some point within $\tlleft t_1, t_2 \tlright$. For ease of notation, let $\windowlength = t_2 - t_1 + 1$ be the cardinality of $\tlleft t_1, t_2 \tlright$. One common method to encode this constraint is via disjunctions~\cite{Kurtz2022a, Cardona2023}, which is given by
\begin{align}
    \varphi_1 \until{\tlleft t_1, t_2 \tlright} \varphi_2 = \bigvee_{t=k+t_1}^{k+t_2-1} \left(\xi^b_{\varphi_2, t+1} \wedge \bigwedge_{t'=k+t_1}^{t} \xi^b_{\varphi_1, t'}  \right) \, ,
\end{align}
to enumerate every admissible state trajectory and constrain the system to one such trajectory. However, this requires additional propositions to be created to track the validity of each admissible trajectory. Since each proposition results in a new binary variable in the representation, the complexity of this method can quickly become untenable. Therefore, we present an alternative way to implement the ``until" constraints as a conjunction of linear inequalities that leads to much lower representational complexity.

\begin{proposition}
    Given a map as described in Sec.~\ref{sec:map_construction} with propositions $\varphi_1$ and $\varphi_2$ where $\xi^b_{\varphi, k} = 1\Leftrightarrow \bm{\sigma} \models_k \varphi$, the MTL specification $\varphi = \varphi_1 \until{\tlleft t_1, t_2 \tlright} \varphi_2$ can be represented with $\windowlength$ linear inequality constraints of the form
\begin{subequations} \label{eq:until-ineq}
\begin{align} 
    &\sum_{j=k+t_1}^{k+t_2} \xi^b_{\varphi_2, j} \geq 1\,, \label{eq:until-event}\\
    \begin{split}\xi^b_{\varphi_2,t} &\leq \sum_{j=k+t_1}^{t-1} \frac{\xi^b_{\varphi_1, j}}{t-t_1-k} + \xi^b_{\varphi_2, j}\,,\\
    \forall t&\in \tlleftopen k+t_1, k+t_2 \tlright\,.\end{split}\label{eq:until-subsq}
\end{align}
\end{subequations}
\end{proposition}
\begin{proof}
    From~\eqref{eq:until-event} and the fact that $\xi^b_{\varphi_2, j}\in\{0,1\}$, $\varphi_2$ must become true at some point in the interval $\tlleft k+t_1, k+t_2 \tlright$. 
    First, consider the case where $\varphi_2$ first becomes true at $t=k+t_1$. Equation~\eqref{eq:until-event} is satisfied as $\sum_{j=k+t_1}^{k+t_2} \xi^b_{\varphi_2, j} \geq \xi^b_{\varphi_2, k+t_1} = 1$, and equation~\eqref{eq:until-subsq} is satisfied as $\sum_{j=k+t_1}^{t-1} \frac{1}{t-t_1-k}\xi^b_{\varphi_1, j} + \xi^b_{\varphi_2, j} \geq \xi^b_{\varphi_2, k+t_1} = 1 \geq \xi^b_{\varphi_2,t}$.
    Then, consider the case where $\varphi_2$ becomes true at some time step $i\in\tlleftopen k+t_1, k+t_2 \tlright$, implying that $\xi^b_{\varphi_2, i} =1$ and $\xi^b_{\varphi_2, t} = 0$ $\forall t \in \tlleft k+t_1, i \tlrightopen$ (similarly satisfying~\eqref{eq:until-event}). Therefore, $\forall t \in \tlleft k+t_1, i \tlrightopen$, \eqref{eq:until-subsq} yields that $0 \leq \sum_{j=k+t_1}^{t-1} \frac{1}{t-t_1-k+1} \xi^b_{\varphi_1, j}$, which is always true since $\xi^b_{\varphi_1, j}\in\{0,1\}$. For $t=i$, \eqref{eq:until-subsq} yields that $1 \leq \sum_{j=k+t_1}^{t-1} \frac{1}{t-t_1-k+1} \xi^b_{\varphi_1, j}$, which is only true in the case that $\xi^b_{\varphi_1, j} = 1$ $\forall j \in \tlleft k+t_1, t \tlrightopen$ (i.e., $\varphi_1$ has been true during the entire interval $\tlleft k+t_1, i \tlrightopen$). Then, $\forall t\in \tlleftopen i, k+t_2 \tlright$, \eqref{eq:until-subsq} holds true as $\xi^b_{\varphi_2,t} \in \{0,1\}$ and $\sum_{j=k+t_1}^{t-1}\xi^b_{\varphi_2,j} \geq 1$ since $\xi^b_{\varphi_2, i} =1$. Thus the desired ``until" behavior is achieved.
\end{proof}

Similar to Proposition~\ref{prop:boolean-hrep}, we construct H-rep polytopes to contain the feasible space of the until operator.
\begin{proposition}
    The feasible space of the constraints provided for the until operator in~\eqref{eq:until-ineq} can be represented as an H-rep polytope of the form
\begin{subequations}\label{eq:until-cz}
\begin{align} 
    \mathcal{H}_{\until{}} &= \left<[L_1\,, L_2], \begin{bmatrix}
        \mathbf{0_{\windowlength-1 \times 1}} \\ -1
    \end{bmatrix}, [\,], [\,]\right>_h \subset \real^{2 \windowlength}\,,
\end{align}
\end{subequations}
such that $[\xi^b_{\varphi_1, k+t_1}, \dots, \xi^b_{\varphi_1, k+t_2}, \xi^b_{\varphi_2, k+t_1}, \dots, \xi^b_{\varphi_2, k+t_2}]^T \in \mathcal{H}_{\until{}}$, where
\begin{subequations} \label{eq:until-cz-mats}
\renewcommand{\arraystretch}{1.1} \begin{align}
    L_1 &= \begin{bmatrix}
        -1 & 0 & 0 & \cdots & 0 & 0\\
        -\frac{1}{2} &  -\frac{1}{2} & 0 & \cdots & 0 & 0\\
        -\frac{1}{3} &  -\frac{1}{3} & -\frac{1}{3} & \cdots & 0 & 0\\
        \vdots & \vdots & \vdots & \ddots & \vdots & \vdots\\
        \frac{-1}{\windowlength-1} & \frac{-1}{\windowlength-1} &\frac{-1}{\windowlength-1} & \cdots & \frac{-1}{\windowlength-1} & 0\\
        0 & 0 & 0 & \cdots & 0 & 0
    \end{bmatrix}\,, \renewcommand{\arraystretch}{1.0}\\
    L_2 &= \begin{bmatrix}
        -1 & 1 & 0 & 0 & \cdots & 0 \\ 
        -1 & -1 & 1 & 0 & \cdots & 0 \\
        -1 & -1 & -1 & 1 & \cdots & 0 \\
        \vdots & \vdots & \vdots & \ddots & \ddots & \vdots\\ 
        -1 & -1 & -1 & \cdots & -1 & 1 \\
        -1 & -1 & -1 & \cdots & -1 & -1
    \end{bmatrix}\,.
\end{align}
\end{subequations}
\end{proposition}

\begin{proof}
    By inspection of the structure in~\eqref{eq:until-cz} and by the definition of the H-rep polytope in~\eqref{eq:hrep}, the first $\windowlength - 1$ constraints enforced by the matrices in~\eqref{eq:until-cz-mats} can be written as
    \begin{align*}
        -\xi^b_{\varphi_1, k+t_1} - \xi^b_{\varphi_2, k+t_1} + \xi^b_{\varphi_2, k+t_1 + 1} &\leq 0\,,\\
        -\sum_{t=k+t_1}^{k+t_1+1}\left(\tfrac{1}{2}\xi^b_{\varphi_1, t} + \xi^b_{\varphi_2, t}\right) + \xi^b_{\varphi_2, k+t_1+2}  &\leq 0\,,\\
        &\vdots \\
        -\sum_{t=k+t_1}^{t_2-1}\left(\frac{\xi^b_{\varphi_1, t}}{\windowlength - 1} + \xi^b_{\varphi_2, t}\right) + \xi^b_{\varphi_2, k+t_2}  &\leq 0\,,
    \end{align*}
    which when rearranged, yields an inequality identical to those in~\eqref{eq:until-ineq}.
\end{proof}

\subsubsection{Eventually and Always}
The eventually operator, $\eventually{\tlleft t_1,t_2 \tlright}, \varphi$ requires that $\varphi$ holds true at least once within the time interval $\tlleft t_1,t_2 \tlright$, and the always operator, $\always{\tlleft t_1,t_2 \tlright} \varphi$, requires that $\varphi$ holds true at \textit{all} times within the interval $\tlleft t_1,t_2 \tlright$. While we have defined these operators in terms of the ``until" operator in~\eqref{eq:eventually-def} and~\eqref{eq:always-def}, the following proposition presents a simpler representation that requires fewer constraints. 
\begin{proposition}
    The MTL formulae for the eventually and always operators, can each be represented with a single linear constraint of the form
\begin{equation} \label{eq:eventually-ineq}
    \bm{\sigma} \models_k \eventually{\tlleft t_1, t_2 \tlright} \varphi \Leftrightarrow \sum_{j=t_1 + k}^{t_2 + k} \xi_{\varphi, j}^b \geq 1\;,
\end{equation}
for the eventually operator and 
\begin{equation}\label{eq:always-ineq}
    \bm{\sigma} \models_k \always{\tlleft t_1, t_2 \tlright} \varphi \Leftrightarrow \sum_{j=t_1+k}^{t_2+k} \xi_{\varphi, j}^b = \windowlength\,,
\end{equation}
for the always operator.
\end{proposition}
\begin{proof}
    From the definition of the eventually operator in~\eqref{eq:mtl-def}, there must be at least one time index in the interval $\tlleft k+t_1, k+t_2 \tlright$ where $\varphi$ holds true. Therefore, there must exist $i \in \tlleft k+t_1, k+t_2 \tlright$ s.t. $\xi^b_{\varphi, i} = 1$, implying that $\sum_{j=t_1 + k}^{t_2 + k} \xi_{\varphi, j}^b \geq 1$. Conversely, $\sum_{j=t_1 + k}^{t_2 + k} \xi_{\varphi, j}^b \geq 1$ with the fact that $\xi^b_{\varphi,j} \in \{0,1\}$ implies that $\exists i \in \tlleft k+t_1, k+t_2 \tlright$ s.t. $\xi^b_{\varphi, j} = 1$. Therefore, $\exists i \in \tlleft k+t_1, k+t_2 \tlright$ s.t. $\bm{\sigma} \models_i \varphi$ and $\eventually{\tlleft t_1,t_2\tlright} \varphi$. Similarly, from the definition of the always operator in~\eqref{eq:mtl-def}, $\varphi$ must hold true for all time indices in the interval $\tlleft k+t_1, k+t_2 \tlright$. That is, $\forall i\in\tlleft k+t_1, k+t_2 \tlright$, $\xi^b_{\varphi, i} = 1$. Therefore, $\sum_{j=t_1+k}^{t_2+k} \xi_{\varphi, j}^b = t_2 - t_1 + 1 = \windowlength$. Conversely, $\sum_{j=t_1+k}^{t_2+k} \xi_{\varphi, j}^b = \windowlength$ implies that $\forall i \in \tlleft k+t_1, k+t_2 \tlright$, $\xi^b_{\varphi, j} = 1$. Therefore, $\forall i \in \tlleft k+t_1, k+t_2 \tlright$, $\bm{\sigma} \models_i \varphi$ and $\always{\tlleft t_1,t_2\tlright} \varphi$.
\end{proof}
The feasible space of the eventually and always operators are again represented using H-rep polytopes.
\begin{proposition}
    The feasible space of the constraints provided for the eventually and always operators in~\eqref{eq:eventually-ineq} and~\eqref{eq:always-ineq} can be represented as H-rep polytopes of the form
    \begin{subequations}
    \begin{align}
        \mathcal{H}_{\eventually{}} &= \left<-\mathbf{1}_{1\times \windowlength}, -1, [\,], [\,]\right>_h \subset \real^{\windowlength}\label{eq:eventually-cz}\,,\\
        \mathcal{H}_{\always{}} &= \left<[\,], [\,], \mathbf{1}_{1\times \windowlength}, \windowlength\right>_h\label{eq:always-cz}\subset \real^{\windowlength}\,,
    \end{align}
    \end{subequations}
    such that $[\xi^b_{\varphi, k+t_1}, \dots, \xi^b_{\varphi, k+t_2}]^T = \bm{\xi} \in \mathcal{H}_{\eventually{}}$,~$\mathcal{H}_{\always{}}$.
\end{proposition}
\begin{proof}
    From the definition of the H-rep polytope in~\eqref{eq:hrep}, \eqref{eq:eventually-cz} can be rewritten as
    \begin{align}
        \mathcal{H}_{\eventually{}} &= \left\{\bm{\xi} \in \real^{\windowlength}\,\middle|\,-\xi_1-\dots-\xi_{\windowlength} \leq -1\right\}\,,
    \end{align}
    which enforces~\eqref{eq:eventually-ineq}, thus constructing the feasible space for the binary generators associated with the eventually operator. The constrained zonotope in~\eqref{eq:always-ineq} can be similarly rewritten as
    \begin{align}
        \mathcal{H}_{\always{}} &= \left\{\bm{\xi} \in \real^{\windowlength}\,\middle|\,\xi_1+\dots+\xi_{\windowlength} = \windowlength\right\}\,,
    \end{align}
    where the linear equality constraint enforces~\eqref{eq:always-ineq}, thus constructing the feasible space for the binary generators associated with the always operator.
\end{proof}

\subsection{Collecting MTL Constraints}
Once the feasible space of a given MTL specification has been constructed and the forward reachable sets have been computed, we can build the feasible space of the lifted system that satisfies both the system dynamics and the MTL specification via set intersections.
\begin{theorem} \label{theorem:mtl-intersection}
    Given a lifted set of forward reachable sets $\mathcal{Z}_{\text{reach}}$ as constructed in Sec.~\ref{sec:reach}, a set $\mathcal{H_\varphi}$ representing the feasible space of specification $\varphi$, and matrix $R\in \{0,1\}^{n^{\mathcal{H}_{\varphi}} \times n^{\mathcal{Z}_{\text{reach}}}}$ defined as
    \begin{equation} \label{eq:R-def}
    R_{i,j} = \left\{ \begin{tabular}{l l}
        1, & if the $i^{\text{th}}$ dimension of $\mathcal{H}_{\varphi}$ corresponds\\
        &\hphantom{if }to the $j^{\text{th}}$ dimension of $\mathcal{Z}_{\text{reach}}$,\\
        0, & \text{otherwise},
    \end{tabular} \right.
    \end{equation}
    a lifted set $\mathcal{Z}_{\varphi}$ of forward reachable sets which satisfy specification $\varphi$ can be constructed through the set intersection
    \begin{equation}
        \mathcal{Z}_{\varphi} = \mathcal{Z}_{\text{reach}} \cap_{R} \mathcal{H}_{\varphi}\,.
    \end{equation}
\end{theorem}
\begin{proof}
    From the definition of the generalized intersection in~\eqref{eq:set-ops-intersection}, $\mathcal{Z}_\varphi$ can be equivalently expressed as $\mathcal{Z}_\varphi = \left\{\z \in \mathcal{Z}_{\text{reach}}\middle| R \z \in \mathcal{H}_\varphi\right\}$ where $R \z$ yields the projection of a point $\z \in \mathcal{Z}_{\text{reach}}$ onto the dimensions associated with $\mathcal{H}_\varphi$. Therefore, $\mathcal{Z}_\varphi$ yields the set which belongs to $\mathcal{Z}_{\text{reach}}$ which also satisfies specification $\varphi$.
\end{proof}
\begin{corollary}
    For an MTL specification expressed in conjunctive normal form $\varphi = \bigwedge_{j=1}^{n_\varphi} \varphi_j$, we can repeatedly apply set intersections to enforce all specifications as
    \begin{equation}
        \mathcal{Z}_{\text{feas}} = \mathcal{Z}_{\text{reach}} \bigcap_{j=1}^{n_\varphi} {}_{R_j} \mathcal{H}_{\varphi_j}\,.\label{eq:reach-mtl-intersect}
    \end{equation}
\end{corollary}
\begin{proof}
    For specification $\varphi = \bigwedge_{j=1}^{n_\varphi} \varphi_j$, repeated application of Theorem~\ref{theorem:mtl-intersection} with $\mathcal{Z}_{\varphi_0} = \mathcal{Z}_{\text{reach}}$ yields
    \begin{subequations}
        \begin{equation}
            \mathcal{Z}_{\varphi_{j}} = \mathcal{Z}_{\varphi_{j-1}} \cap_{R_{j}} \mathcal{H}_{\varphi_{j}}\,, \forall j\in\{1,\dots,n_\varphi\}\,.
        \end{equation}
    \end{subequations}
    A trajectory in $\mathcal{Z}_{\varphi_{n_\varphi}}$ is guaranteed to satisfy each specification $\varphi_j$ as $\mathcal{Z}_{\varphi_{n_\varphi}} \subseteq \mathcal{Z}_{\text{reach}} \cap_{R_j} \mathcal{H}_{\varphi_j}$, $\forall j \in \{1,\dots,n_\varphi\}$.
\end{proof}

The additional complexity in the set representation incurred by each intersection in~\eqref{eq:reach-mtl-intersect} is summarized in Table~\ref{tab:added-complexity}. For all operations except ``until", up to one constraint and one continuous generator are introduced. For ``until", an additional $\tau$ continuous generators and $\tau$ constraints are introduced. For all operators, no additional binary generators are required. 
\begin{table}[]
    \centering
    \caption{Added representational complexity resulting from intersections of a hybrid zonotope, $\mathcal{Z}_{\text{\MakeLowercase{reach}}}$, with H-rep polytope representations of MTL operators in~\eqref{eq:reach-mtl-intersect}.}
    \begin{tabular}{|c|c|c|c|}
         Operator & \multicolumn{3}{c}{Additional Complexity} \vline\\
         & $n_g$ & $n_b$ & $n_c$ \\ \hline
         $\pi$, $\neg$, $\wedge$ & $0$ & $0$ & $1$\\
         $\vee$ & $1$ & $0$ & $1$\\
         $\until{\tlleft t_1, t_2 \tlright}$ & $\windowlength$ & $0$ & $\windowlength$\\
         $\eventually{\tlleft t_1, t_2 \tlright}$ & $1$ & $0$ & $1$\\
         $\always{\tlleft t_1, t_2 \tlright}$ & $0$ & $0$ & $1$
    \end{tabular}
    \label{tab:added-complexity}
\end{table}

\subsection{Optimization Formulation}
To capture region-dependent disturbances, we modify~\eqref{eq:generic-lqr} as
\begin{subequations}
\begin{align}
    \u^* = \argmin_{\u}& \norm{\x_N}^2_{Q_N} + \sum_{k=0}^{N-1} \norm{\x_k}^2_Q + \norm{\u_k}^2_R\,, \\
    \text{s.t. }& \begin{bmatrix}
        \x_0 \\ \bm{\xi}^b_0
    \end{bmatrix} \in \mathcal{X}'_0\,, \; \begin{bmatrix}
        \x_k \\ \bm{\xi}^b_k
    \end{bmatrix} \in \mathcal{M}'_k\,,\; \u_k \in \mathcal{U}\,,\\
    &\x_{k+1}  = A \x_k + B \u_k + W \bm{\xi}^b_k \,,\\
    &\x \models_0 \varphi \,.
\end{align}
\end{subequations}
Generating a trajectory of optimal states and inputs in the non-lifted space that satisfies the MTL specification is equivalent to finding an optimal point within the lifted hybrid zonotope, $\mathcal{Z}_{\text{feas}}$. Therefore, we generate an MIQP of the form
\begin{subequations}\label{eq:miqp-lifted}
\begin{align}
    \z^* &= \argmin_{\z} \norm{\z}^2_P\,,\\
    \text{s.t. } \z &\in \mathcal{Z}_{\text{feas}}\,,
\end{align}
\end{subequations}
where
\begin{equation}
    P = \mathtt{blkdiag}\left(Q,\mathbf{0},R,Q,\mathbf{0},\dots,R, Q_N, \mathbf{0}\right)\,.
\end{equation}
The optimization problem in~\eqref{eq:miqp-lifted} can be solved by a variety of MIQP solvers, such as Gurobi~\cite{gurobi} or SCIP~\cite{scip}.
\section{Simulation Examples} \label{sec:examples}
This section presents three numerical examples to demonstrate the proposed motion planning framework. 
The first example is an existing problem from the literature where an autonomous agent must obtain keys to open doors and reach a goal region.
The proposed method is benchmarked with the methods in~\cite{Kurtz2022}. 
The second example is inspired by the traveling salesperson problem and an example from~\cite{Kurtz2022}, where the agent must visit moving targets. The third example focuses on energy-aware autonomy, where the agent must visit charging stations to replenish an energy state and is subject to region-dependent disturbances in its motion states. 
All numerical examples were implemented on a desktop computer with an Intel\textsuperscript{\textregistered} Core\textsuperscript{\texttrademark} i7-14700 processor with 28 parallelized threads and 32GB of RAM, running Ubuntu 24.04. The optimization problems are of the form in~\eqref{eq:miqp-lifted}, constructed using ZonoOpt~\cite{robbins2025sparsitypromotingreachabilityanalysisoptimization}, and solved with Gurobi~\cite{gurobi} with a relative tolerance of 1\% (unless otherwise stated).

\subsection{Door-key Problem} \label{sec:example-door-key}
The first example is modified from~\cite[Example 4]{Kurtz2022} and used to benchmark the proposed approach. In this scenario, the agent must reach a goal region, $\mathcal{G}$. However, two door regions, $\mathcal{D}_1$ and $\mathcal{D}_2$, block a corridor to the goal and cannot be entered until the agent collects the appropriate keys, $\mathcal{K}_1$ and $\mathcal{K}_2$. The goal, door, and key regions are represented by green, red, and blue rectangles in Fig.~\ref{fig:dk25}, respectively. The gray regions in Fig.~\ref{fig:dk25} represent keep-out zones. The system model is of the form~\eqref{eq:lti} with states $\x = [x, y, \dot{x}, \dot{y}]^T$ where $x$ and $y$ are the position states, inputs $\u = [\ddot{x}, \ddot{y}]^T$, system matrices
\begin{equation} \label{eq:dbl-int-dyn-simple}
    A = \begin{bmatrix}
        1 & 0 & 1 & 0 \\
        0 & 1 & 0 & 1 \\
        0 & 0 & 1 & 0 \\
        0 & 0 & 0 & 1         
    \end{bmatrix},\; B = \begin{bmatrix}
        0 & 0\\
        0 & 0\\
        1 & 0\\
        0 & 1
    \end{bmatrix},\; \w_k = \mathbf{0}\,,
\end{equation} 
and cost matrices
\begin{equation}
    Q = Q_N = \mathtt{diag}\left([0,0,1,1]\right)\,,\;R = I_2\,. \label{eq:Q-R-door-key}
\end{equation} 
This problem is slightly modified from~\cite{Kurtz2022} to require that the agent reach $\mathcal{G}$ at the \textit{end} of its horizon. The door-key problem can be described by the MTL specification
\begin{equation} \label{eq:spec-door-key}
    \varphi_{A} = \eventually{\tlleft N\tlright} \mathcal{G} \wedge \left(\bigwedge_{i=1}^2 \neg \mathcal{D}_i \until{\tlleft0,N\tlright} \mathcal{K}_i\right)\,.
\end{equation}

We use a combination of a DCP and an NCP to represent the non-convex free space, doors, goal, and key regions as a hybrid zonotope with 24 continuous generators, 10 binary generators, and 13 constraints. The door regions require a DCP to enforce that the agent may not enter until it has visited the key regions. The remaining non-convex free space can be constructed using an NCP as there are no further entry restrictions for these regions. In Fig.~\ref{fig:dk25}, the convex polytopes which tile the background of this example are shown as the hatched regions.
The proposed method produces visually identical trajectories as when the robustness terms in the cost function of~\cite{Kurtz2022} are disabled. Fig.~\ref{fig:dk25} shows the solution for $N=25$ with Gurobi's relative optimality gap parameter set to 1\%. 

The representational complexity and solution time for several values of $N$ and two values for the relative optimality gap are summarized in Fig.~\ref{fig:comparison-door-key} and Table~\ref{tab:comparison-door-key}. An optimality gap of 50\% is included to illustrate how quickly a ``reasonable" solution can be found; this is particularly useful when routinely performing these calculations online with restricted computation time. Furthermore, this optimality gap is based on what can be proven by the solver. The true sub-optimality of a given solution may be much less than that the solver's verified optimality gap prior to convergence.  Fig.~\ref{fig:comparison-door-key} shows
the number of binary variables \textit{after} Gurobi's default pre-solve routine. 
For all cases, the proposed method results in an optimization problem with roughly an order of magnitude fewer binary variables than the methods in~\cite{Kurtz2022}, which uses the big M method. For all but three scenarios, the proposed method has an improved computation time.

\begin{figure}[htb]
    \centering
    \input{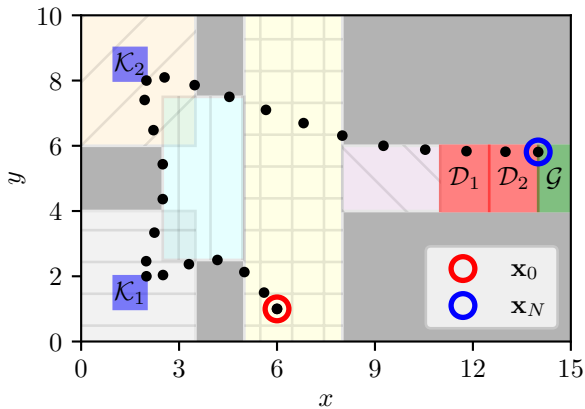}
    \caption{Door-key problem where the agent must reach the goal region, $\mathcal{G}$. However, the two door regions, $\mathcal{D}_1$ and $\mathcal{D}_2$, cannot bet entered until the agent visits the corresponding key regions, $\mathcal{K}_1$ and $\mathcal{K}_2$. The hatched regions show the polytopic partition and the gray regions represent obstacles.}
    \label{fig:dk25}
\end{figure}

\begin{figure}
    \centering
    \input{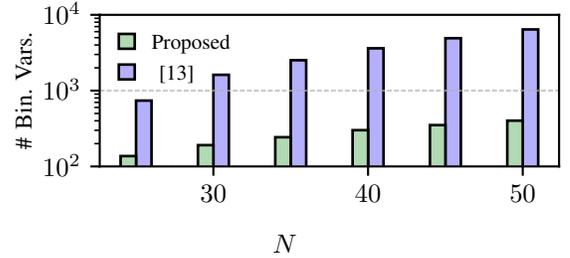}
    \caption{Comparison of problem sizes for the door-key problem.}
    \label{fig:comparison-door-key}
\end{figure}
\begin{table}[htb]
    \centering
    \caption{Comparison of solution times for the door-key problem.}
    \begin{tabular}{|c|c|c|c|c|}
         & \multicolumn{4}{c}{Solution Time [s]} \vline \\ 
         $N$& \multicolumn{2}{c}{50\% Optimal} \vline & \multicolumn{2}{c}{1\% Optimal} \vline \\
         &  Ours & \cite{Kurtz2022} & Ours & \cite{Kurtz2022} \\ \hline
         {25} & {0.85} & \textbf{0.51} & {0.88} & \textbf{0.63}\\
         {30} & \textbf{2.12} & {2.74} & \textbf{4.05} & {5.96}\\
         {35} & \textbf{1.57} & {5.51} & \textbf{5.61} & {10.94}\\
         {40} & \textbf{3.13} & {15.54} & {34.80} & \textbf{31.71}\\
         {45} & \textbf{5.09} & {37.49} & \textbf{151.60} & {166.27}\\
         {50} & \textbf{13.71} & {85.40} & \textbf{112.37} & {2696.67}
    \end{tabular}
    \label{tab:comparison-door-key}
\end{table}

\subsection{Traveling Salesperson Problem} \label{sec:example-tsp}
The second numerical example is inspired by~\cite[Example 3]{Kurtz2022} and the traveling salesperson problem. 
There are targets of five different colors with two regions per color, $\mathcal{C}_i$, as shown in Fig.~\ref{fig:tsp-dvd}. To demonstrate the ability to accommodate time-varying maps, each target moves at a fixed velocity and rebounds off the boundaries of map. The objective is to visit at least one region of each color. 
We use the same dynamic model and cost matrices as in Sec.~\ref{sec:example-door-key}.

The required behavior can be described with the MTL specification
\begin{equation}
    \varphi_{B} = \bigwedge_{i=1}^5 \Diamond_{\tlleft0,N\tlright} \mathcal{C}_i\,.
\end{equation}
This example uses an NCP as there are no keep-out zones or areas with restricted entry, resulting in a hybrid zonotope map with 12 continuous generators, 11 binary generators, and 7 constraints at each time step. For $N=25$, the MIP uses 192 continuous variables, 234 binary variables, and takes approximately 0.75 seconds to compute. The resulting trajectory is shown in Fig.~\ref{fig:tsp-dvd}. 

\begin{figure*}[htb]
    \centering
    \input{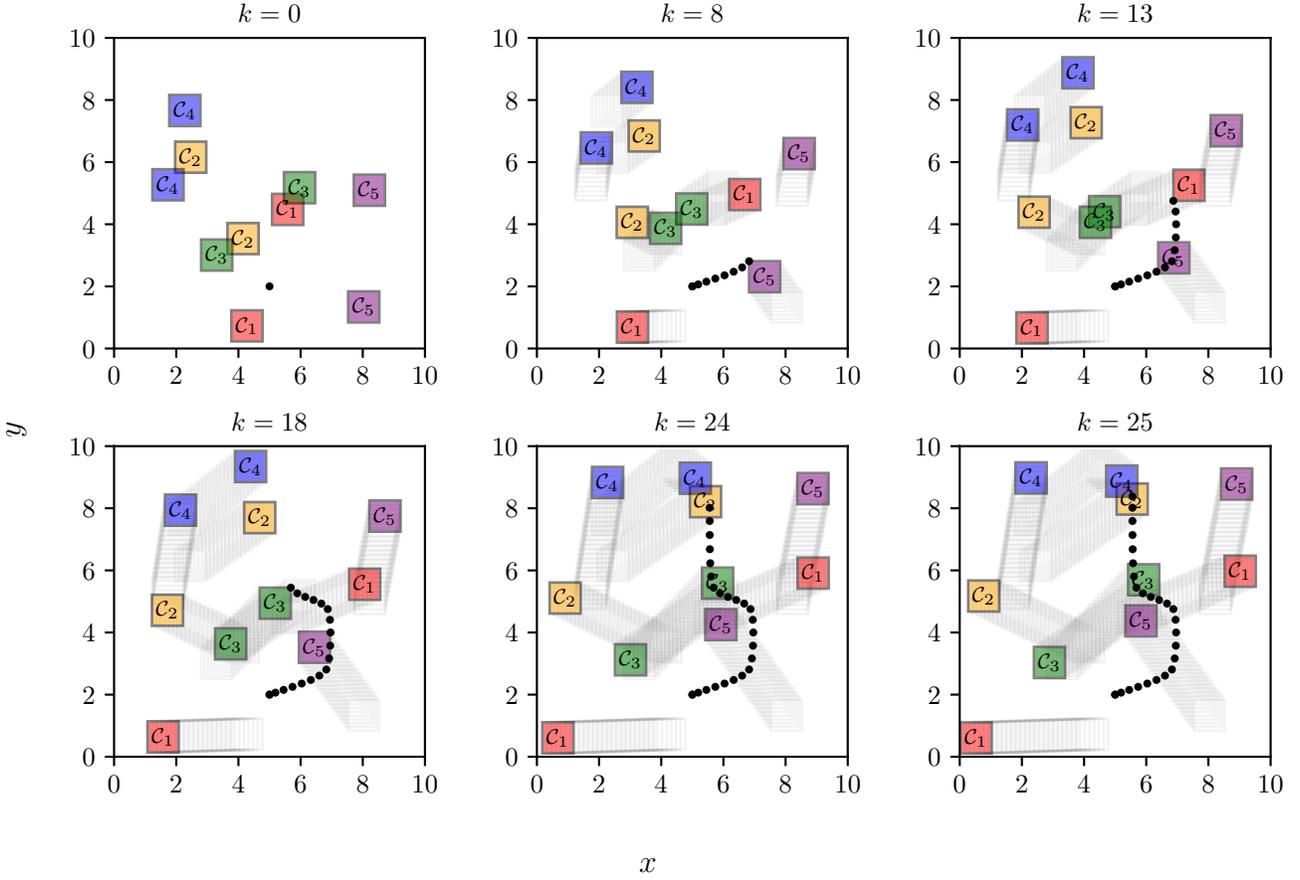}
    \caption{Modified traveling salesperson problem where the agent must visit at least one region of each color. The colored regions are moving with constant velocity but rebound off the boundaries of the map. The agent visits a purple region at $k=8$, a red region at $k=13$, a green region at $k=18$, a yellow region at $k=24$, and a blue region at $k=25$.}
    \label{fig:tsp-dvd}
\end{figure*}

\subsection{Energy-aware Problem} \label{sec:example-energy}
The third example focuses on the context of energy-aware autonomy. The autonomous vehicle is assumed to be an aerial drone and must navigate to a collection of target sets while subject to region-dependent wind disturbances and energy constraints. The system model is similar to that of the previous examples; however, the state vector in~\eqref{eq:dbl-int-dyn-simple} is augmented with an energy state, $e$, representing a state of charge (SOC). Therefore, the system has states $\x = [x, y, \dot{x}, \dot{y}, e]^T$, inputs $\u = [\ddot{x}, \ddot{y}]^T$, and system matrices
\begin{equation}
    A = \begin{bmatrix}
        1 & 0 & 1 & 0 & 0\\
        0 & 1 & 0 & 1 & 0 \\
        0 & 0 & 1 & 0 & 0\\
        0 & 0 & 0 & 1 & 0\\
        0 & 0 & 0 & 0 & 1
    \end{bmatrix},\; B = 
    \begin{bmatrix}
        0 & 0\\ 0 & 0\\ 1 & 0\\ 0 & 1\\ 0 & 0
    \end{bmatrix},\; \w_k = 
    \begin{bmatrix}
        W \\ \mathbf{0}_{2\times n_b} \\\mathbf{d}^T
    \end{bmatrix} \bm{\xi}^b_k,
\end{equation}
where $W = [\windvec_1, \windvec_2, \dots, \windvec_{30}]$ and $\mathbf{d}^T = [d_1, d_2, \dots, d_{30}]$ capture region-dependent disturbances representing the wind and power effects, respectively. These dependencies are enforced through a linear transformation of the binary variables associated with the map, enabled by the augmentation in~\eqref{eq:hz-augment}. Each wind vector $\windvec_i\in\real^2$ has a magnitude of $0.2$ and points nominally perpendicular to the center of the map in a counter clockwise fashion with a $\pm15^\circ$ random variation. There are three regions that must be visited, $\mathcal{G}_i$, four keep-out areas, $\mathcal{O}_i$, and two regions that recharge the energy state of the system at a rate of $20$\% per time step ($d_i = 0.2$), $\mathcal{C}_i$. Every region that is not one of the charging regions depletes the SOC at a rate of $10$\% per time step ($d_i = -0.1$). The SOC is constrained to not fully deplete or overcharge, i.e., $0 \leq e_k \leq 1$. 

We consider two scenarios that show how the optimized trajectory changes when a restricted arrival window is applied to one of the targets. In the first scenario, the agent may visit the targets in any order within the $25$ step horizon. The MTL specification describing this requirement is
\begin{equation} \label{eq:mtl-c1}
    \varphi_{C1} = \left(\bigwedge_{i=1}^3 \eventually{\tlleft0,25\tlright} \mathcal{G}_i\right) \wedge \left(\bigwedge_{i=1}^4 \always{\tlleft0,25\tlright} \neg \mathcal{O}_i\right) \,.
\end{equation}
In the second scenario, the third target must be visited within the first $12$ time steps, which could indicate a high-priority or time-sensitive objective. The horizon for this scenario is extended to $35$ steps to accommodate the increased path length of the agent. This requirement is described by the MTL specification
\begin{equation} \label{eq:mtl-c2}
    \varphi_{C2} = \left(\bigwedge_{i=1}^2 \eventually{\tlleft0,35\tlright} \mathcal{G}_i\right) \wedge \eventually{\tlleft0,12\tlright} \mathcal{G}_3 \wedge \left(\bigwedge_{i=1}^4 \always{\tlleft0,35\tlright} \neg \mathcal{O}_i\right) \,.
\end{equation}

This example uses a DCP, resulting in a hybrid zonotope map with $12$ continuous generators, $30$ binary generators, and $7$ constraints. The overall MIP uses 175 continuous variables and 439 binary variables for the first scenario, and 245 continuous variables and 698 binary variables for the second scenario.
The first scenario takes an approximately $12.4$ seconds to compute and the second scenario takes approximately $21.5$ seconds to compute. The resulting trajectories are shown in Fig.~\ref{fig:wind}. Under~\eqref{eq:mtl-c1}, the agent visits $\mathcal{G}_2$, $\mathcal{G}_3$, then $\mathcal{G}_1$ while recharging between goal regions. Under~\eqref{eq:mtl-c2}, there is not enough time to visit $\mathcal{G}_2$ first, so the agent must circle back through the middle corridor.

\begin{figure*}[htb]
    \centering
    \input{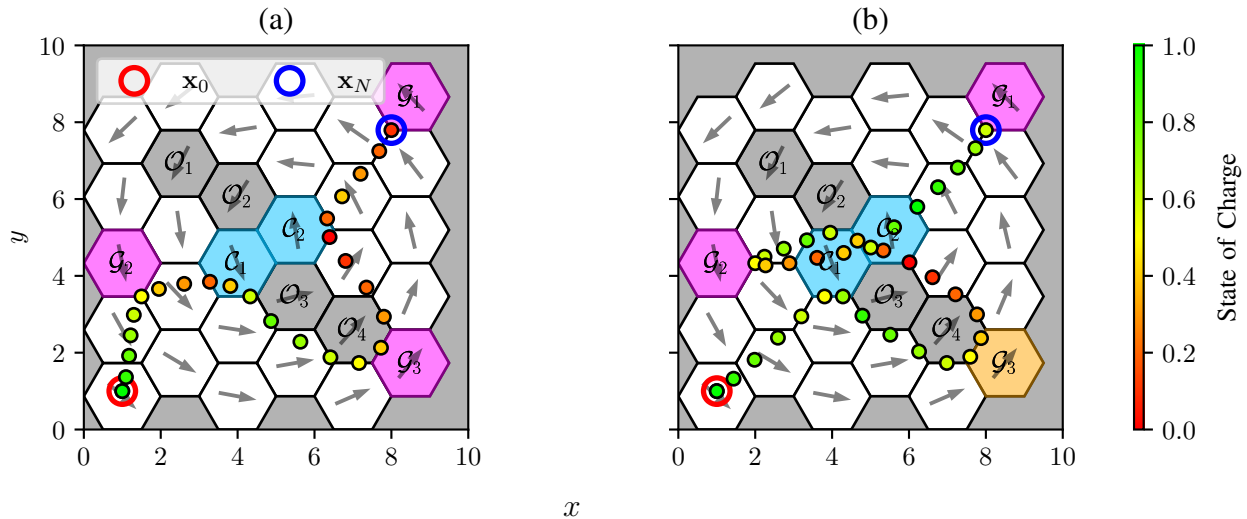}
    \caption{Energy-aware problem where the agent must visit three goal regions, $\mathcal{G}_i$. In (a), with no timing constraints on when regions are visited, the agent first visits $\mathcal{G}_2$, recharges, visits $\mathcal{G}_3$, recharges again, and then visits $\mathcal{G}_1$. In (b), the agent must visit $\mathcal{G}_3$ within the first 12 time steps and does not have enough time to visit $\mathcal{G}_2$ on its way to $\mathcal{G}_3$, so it must circle back.}
    \label{fig:wind}
\end{figure*}
\section{Experimental Application} \label{sec:experiment}
This section describes an experimental application of proposed method. This was implemented using a combination of C++ and Python within a ROS2 framework on a computer with identical specifications as in the previous section. The optimization problem in~\eqref{eq:miqp-lifted} was again constructed in ZonoOpt, but solved using CVXPY~\cite{cvxpy} with SCIP~\cite{scip} as the underlying solver.
The test platforms are a Husarion ROSbot 3 and a Clearpath Robotics Turtlebot 4, both of which are small differential drive robots. Optitrack Prime$^\text{X}$22 motion capture cameras were used to provide state feedback. 
\subsection{Mission Scenario}
This experiment is inspired by last mile delivery problems, which concern the delivery of packages in the final leg of a supply chain.
Specifically, two agents must coordinate their movement to deliver packages to four locations. One agent represents a slow-moving cargo truck, which can carry all four packages. The second agent represents a fast-moving delivery vehicle that can only carry one package at a time. The delivery vehicle must rendezvous with the cargo truck to pick up each package before delivering it to one of the locations. The platform used for the delivery vehicle is the ROSbot 3 and the platform used for the cargo truck is the Turtlebot 4, as the Turtlebot has a lower maximum speed than the ROSbot.
\subsection{Vehicle Models}
Each vehicle is modeled by unicycle dynamics with first-order speed and turn rate inputs as
\begin{equation} \label{eq:unicycle}
    \dot{x} = v \cos(\theta)\,,\; \dot{y} = v \sin(\theta)\,,\; \dot{\theta} = \omega\,,
\end{equation}
where $x$ and $y$ are the position states, $\theta$ is the heading angle, and $v$ and $\omega$ are the linear and angular velocity inputs, respectively. It has been shown that~\eqref{eq:unicycle} accurately describes differential drive robots~\cite{becker2014controlling}. The unicycle model is differentially flat with respect to $x$ and $y$~\cite{diffflat}, enabling planning with a double integrator model of the form~\eqref{eq:lti} with states $\x= [x, y, \dot{x}, \dot{y}]^T$, inputs $\u = [\ddot{x}, \ddot{y}]$, and system matrices 
\begin{align}
    A &= \begin{bmatrix}
        1 & 0 & \Delta t & 0 \\
        0 & 1 & 0 & \Delta t \\
        0 & 0 & 1 & 0 \\
        0 & 0 & 0 & 1 \\
    \end{bmatrix}\,,
    &B = \begin{bmatrix}
        \tfrac{1}{2}\Delta t^2 & 0 \\
        0 & \tfrac{1}{2}\Delta t^2 \\
        \Delta t & 0 \\
        0 & \Delta t \\
    \end{bmatrix}\,,
    & \w_k = \mathbf{0}\,,
\end{align}
where $\Delta t$ is the discretization period, set to 1.5 seconds in this experiment. The states and inputs of the nonlinear model in~\eqref{eq:unicycle}  can be reconstructed from trajectories of the linear system as
\begin{equation}
    v = \sqrt{\dot{x}^2 + \dot{y}^2}\,,\; \theta = \mathtt{atan2}\left(\dot{y},\dot{x}\right)\,,\; \omega = \frac{\dot{x}\ddot{y}-\dot{y}\ddot{x}}{\dot{x}^2 + \dot{y}^2}\,.
\end{equation}
\subsection{Constraints and MTL Specification}
The state and input constraints for each vehicle are defined as
\begin{subequations} \label{eq:vehicleconstraints}
    \begin{align}
        \underline{x} \leq x &\leq \overline{x}\,,\\
        \underline{y} \leq y &\leq \overline{y}\,,\\ 
        \sqrt{\dot{x}^2 + \dot{y}^2} &\leq v_{\max}\,,\label{eq:vmax}\\
        \sqrt{\ddot{x}^2 + \ddot{y}^2} &\leq v_{\min} \omega_{\max}\,, \label{eq:amax}
    \end{align}
\end{subequations}
where $\underline{x}$, $\overline{x}$, $\underline{y}$, and $\overline{y}$ define the boundaries of the rectangular free space, $v_{\max}$ and $\omega_{\max}$ are the vehicle's maximum linear and angular velocities, and $v_{\min}$ is a minimum speed used to conservatively limit the turn rate of the vehicle, as in~\cite{Whitaker2021}. For the cargo truck, $v_{\max} = 0.2 \tfrac{\text{m}}{\text{s}}$, $\omega_{\max} = 1.89\tfrac{\text{rad}}{\text{s}}$, and $v_{\min} = 0.05 \tfrac{\text{m}}{\text{s}}$. For the delivery vehicle, $v_{\max} = 1.0 \tfrac{\text{m}}{\text{s}}$, $\omega_{\max} = 3.14\tfrac{\text{rad}}{\text{s}}$, and $v_{\min} = 0.05 \tfrac{\text{m}}{\text{s}}$. The inequalities in~\eqref{eq:vmax} and~\eqref{eq:amax} are approximated using an inscribed regular hexagon and represented as a zonotope. The boundaries of the testing area are $\underline{x} = 0$ m, $\overline{x} = 6.27$ m, $\underline{y} = 0$ m, and $\overline{y} = 10.62$ m. 

To create a centralized path planner that governs the states of both vehicles and package exchanges, we construct an augmented state vector that contains the kinematic states of each vehicle, $\x_1$ and $\x_2$, a binary state to track whether the delivery vehicle currently possesses a package, $p$, and the relative or residual position of the two vehicles, $\rvec = [x_1 - x_2,\, y_1 - y_2]^T$. That is, $\tilde{\x} = [\x_1^T, \x_2^T, p, \rvec^T]^T$, $\tilde{\u} = [\u_1^T, \u_2^T]^T$, with system matrices
\begin{align}
\begin{split}
    \tilde{A} &= \begin{bmatrix}
        A & \mathbf{0} & \mathbf{0} & \mathbf{0}\\
        \mathbf{0} & A & \mathbf{0} & \mathbf{0}\\
        \mathbf{0} & \mathbf{0} & 0 & \mathbf{0}\\
        H A & -H A & \mathbf{0} & \mathbf{0}
    \end{bmatrix},\; \tilde{B} = \begin{bmatrix}
        B & \mathbf{0}\\
        \mathbf{0} & B\\
        \mathbf{0} & \mathbf{0}\\
        H B & -H B
    \end{bmatrix},\\
    \w_k &= \begin{bmatrix}
        \mathbf{0}_{8 \times n_b}\\
        \mathbf{d}^T\\
        \mathbf{0}_{2 \times n_b}\\
    \end{bmatrix} \bm{\xi}^b_k\,,
\end{split}
\end{align}
where $H = [I_2\; \mathbf{0}_{2\times2}]$ extracts the position of each vehicle from its state vector and $\mathbf{d} = [d_1, \dots, d_j, \dots, d_{n_b}]^T\in\real^{n_b}$ increments $p$ when the relative position between the vehicles is sufficiently close or decrements $p$ when the delivery vehicle visits a delivery location. That is, $d_j = 1$ if $\xi^b_j$ is the generator associated with the package exchange area, $d_j = -1$ if $\xi^b_j$ is the generator associated with the any of the delivery locations, and $d_j = 0$ otherwise.

In addition to the constraints in~\eqref{eq:vehicleconstraints}, the package delivery truck can only carry one package at a time, so $0 \leq p \leq 1$. The constraints on the residual states are defined such that the two vehicles are assumed to be able to exchange packages when the delivery vehicle is within an outer radius, $\overline{r}$, of the cargo truck. However, the delivery vehicle is not allowed within an inner radius, $\underline{r}$, of the truck as to prevent collision. This non-convex constraint set $\mathcal{R}$ is defined as
\begin{equation}
    \mathcal{R} = \mathcal{R}^C_i \cap \left(\left(\mathcal{S} \oplus (-\mathcal{S)}\right)\cup \mathcal{R}_o\right)\,,
\end{equation}
where $\mathcal{R}^C_i$ is the complement of the circumscribed regular hexagon with respect to $\underline{r}$, $\mathcal{S} = [\underline{x}, \overline{x}] \times [\underline{y}, \overline{y}]$, and $\mathcal{R}_o$ is the inscribed regular polygon with respect to $\overline{r}$. This constraint is represented as a hybrid zonotope, using a combination of a DCP and an NCP with 30 continuous generators, 10 binary generators, and 20 constraints. $\mathcal{R}$ is depicted graphically in Fig.~\ref{fig:residual}. In this experiment, an inner radius of 0.4 m and an outer radius of 0.6 m were chosen.

\begin{figure}[htb]
    \centering
    \input{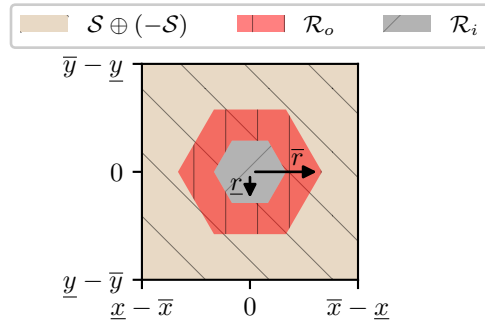}
    \caption{Hybrid zonotope representation of the feasible space for the residual states, $\rvec$. The relative position is not allowed to enter the gray hashed region as this would cause the two vehicles to collide. Package transfer from the cargo truck to the delivery vehicle is allowed when $\rvec$ is inside the red region. The beige region captures the remaining feasible space, allowing the two vehicles to be on opposite corners of the map.}
    \label{fig:residual}
\end{figure}

The MTL specification for this package delivery experiment is
\begin{equation}
    \varphi_P = \bigwedge_{i=1}^4 \eventually{\tlleft 0,N \tlright} \mathcal{G}_i \wedge \eventually{\tlleft N \tlright} \mathcal{X}_N\,,
\end{equation}
where $\mathcal{G}_i$ are the four delivery locations and $\mathcal{X}_N$ is the terminal state constraint of the lifted system. 
\subsection{Results}
The overall MIP for $N=35$ uses 1673 continuous variables and 480 binary variables, takes 82.4 seconds to compute, and yields the mission plan depicted with dashed lines in Fig.~\ref{fig:experiment}.
To conduct the experiment, this  is solved \textit{a priori} and then passed to a closed-loop path-following motion controller for each vehicle, as described in~\cite{robbins2025integration}. The solid lines depict the true path traveled by each vehicle. Fig.~\ref{fig:experiment}(a) shows the initial position of each vehicle. Figs.~\ref{fig:experiment}(b), (c), (d) and (e) show when the first, second, third, and fourth packages are exchanged from the cargo truck to the delivery vehicle, respectively.
Fig.~\ref{fig:experiment}(f) shows the final position of each vehicle.

\begin{figure*}[ht]
    \centering
    \input{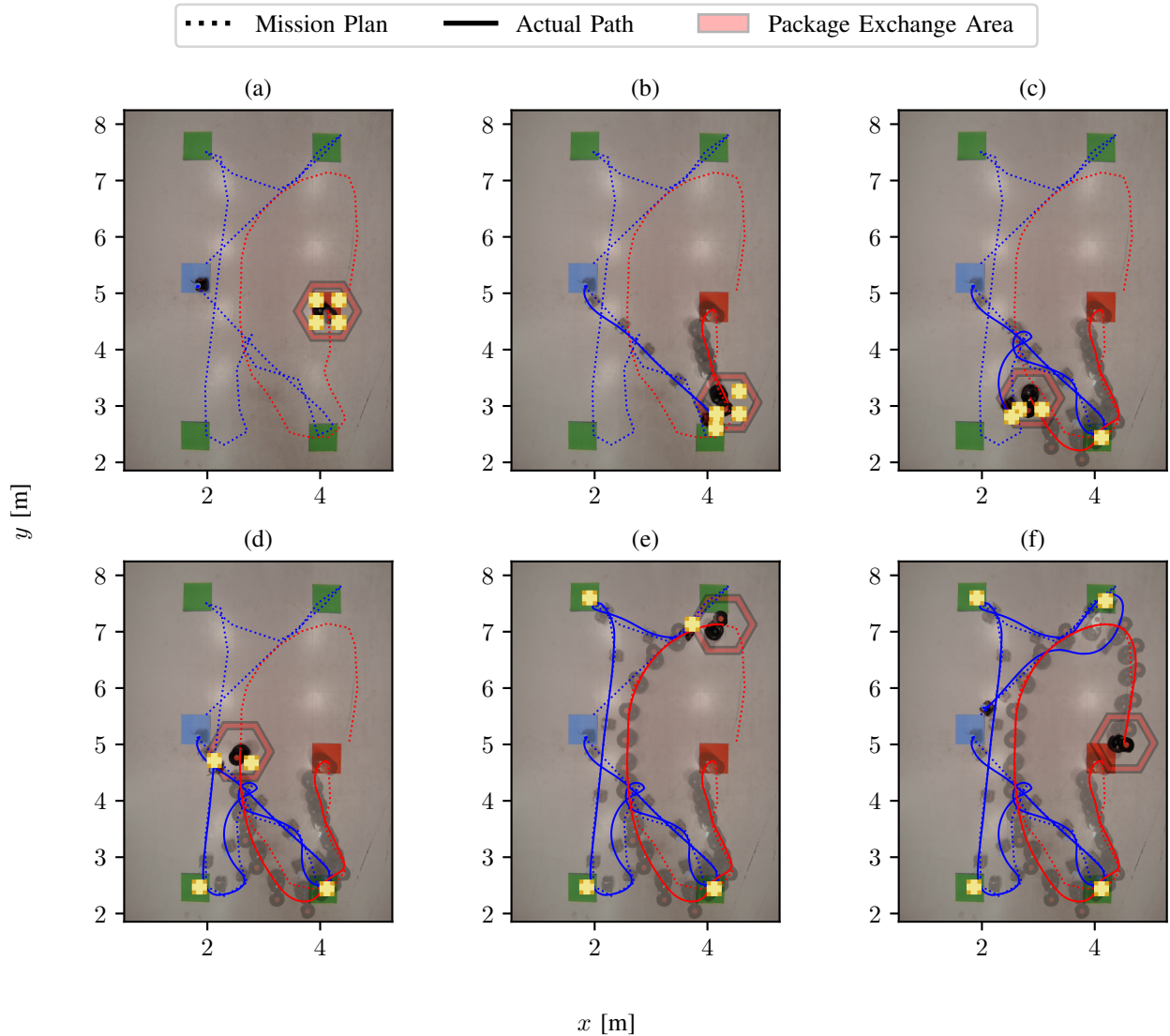}
    \caption{Experimental application of the last mile package delivery problem. The mission plan for each vehicle is shown as a dotted line while the actual path of each vehicle is shown as a solid line. The packages are represented by beige squares, delivery locations are represented by green squares, and the initial and final locations of the vehicles are represented by green and blue squares, respectively.  A video of this experiment can be found at \href{https://www.youtube.com/watch?v=oa1Wx5CMg0c}{\texttt{https://www.youtube.com/watch?v=oa1Wx5CMg0c}}.}
    \label{fig:experiment}
\end{figure*}

\section{Conclusion} \label{sec:conclusion}
This paper explores a method for constructing mixed-integer optimization problems for dynamic systems subject to MTL constraints. A combination of reachability analysis and set-based representations of key MTL operators is leveraged to construct the feasible space of satisfactory trajectories using the hybrid zonotope set representation. When benchmarked against a state-of-the-art method from the literature, the proposed method resulted in an MIP with roughly an order of magnitude fewer binary variables and a reduced computation time in most cases. Numerical examples and an experimental evaluation demonstrate the proposed method's ability to solve motion planning problems that consider time-varying environments, region-dependent disturbances, and coordinated movement of two agents. 

The primary limitation of the proposed approach is that the system can only satisfy one proposition at each time step because of how propositions are related to regions of the map. 
While the agent's position can fall within multiple regions at one time when using an NCP, the proposed representation allows only one region to be ``active'' at a time, based on its associated proposition. Future work should explore alternate ways to partition the map when overlapping regions are present, perhaps by designating the intersection of these regions with a new proposition defined to satisfy both propositions.
\section{Acknowledgements}
This work was supported by the Office of Naval Research under Award N000142512051. Any opinions, findings, and conclusions or recommendations expressed in this material are those of the authors and do not necessarily reflect the views of the Office of Naval Research.

\newpage
\bibliographystyle{ieeetr}
\bibliography{references}

\end{document}